\documentclass[conference]{IEEEtran}

\IEEEoverridecommandlockouts

\usepackage[cmex10]{amsmath}
\usepackage{array}
\usepackage{mdwmath}
\usepackage{mdwtab}
\usepackage{eqparbox}
\usepackage{amsmath}
\usepackage{amssymb,amsbsy}
\usepackage{amsthm}
\usepackage{graphicx}
\usepackage{color}
\usepackage{algorithm}
\usepackage{algpseudocode}

\newcommand{\Fig}[1]{Fig.~\ref{#1}}
\graphicspath{{figures/}}

\newtheorem{lemma}{Lemma}
\newtheorem{theorem}{Theorem}

\newtheorem{remark}{Remark}

\newtheorem{statement}{Statement}

\begin{document}

\sloppy

%% Paper Title
%% You can use linebreaks \\ within to get better formatting as
%% desired. 
\title{On the Capacity of the Multiuser Vector Adder Channel}

%% Author names and affiliations:
%%
%% Avoiding spaces at the end of the author lines is not a problem with
%% conference papers because we don't use \thanks or \IEEEmembership.
%%
%% For several authors with only one affiliation:
%%
% \author{
%   \IEEEauthorblockN{Hui-Ting Chang and Stefan M.~Moser}
%   \IEEEauthorblockA{Department of Electrical and Computer Engineering\\
%     National Chiao Tung University (NCTU)\\
%     Hsinchu, Taiwan\\
%     Email: \{email-of-hui-ting,email-of-stefan\}@ieee.org} 
% }
%%
%% For up to three affiliations:
%%
\author{
  \IEEEauthorblockN{Alexey Frolov, Pavel Rybin and Victor Zyablov}
	
  \IEEEauthorblockA{\small Inst. for Information Transmission Problems\\
    Russian Academy of Sciences\\Moscow, Russia\\
    Email: \{alexey.frolov, prybin, zyablov\}@iitp.ru
  }
}

%%
%% For over three affiliations, or if they all won't fit within the width
%% of the page, use this alternative format:
%%
% \author{
%   \IEEEauthorblockN{
%     Michael Shell\IEEEauthorrefmark{1},
%     Homer Simpson\IEEEauthorrefmark{2},
%     James Kirk\IEEEauthorrefmark{3}, 
%     Montgomery Scott\IEEEauthorrefmark{3} and
%     Eldon Tyrell\IEEEauthorrefmark{4}}
%   \IEEEauthorblockA{
%     \IEEEauthorrefmark{1}School of Electrical and Computer Engineering\\
%     Georgia Institute of Technology, Atlanta, Georgia 30332--0250\\ 
%     Email: see http://www.michaelshell.org/contact.html}
%   \IEEEauthorblockA{
%     \IEEEauthorrefmark{2}Twentieth Century Fox, Springfield, USA\\
%     Email: homer@thesimpsons.com}
%   \IEEEauthorblockA{
%     \IEEEauthorrefmark{3}Starfleet Academy, San Francisco, California 96678-2391\\
%     Telephone: (800) 555--1212, Fax: (888) 555--1212}
%   \IEEEauthorblockA{
%     \IEEEauthorrefmark{4}Tyrell Inc., 123 Replicant Street, Los Angeles, California 90210--4321}
% }

%% Use for special paper notices
%\IEEEspecialpapernotice{(Invited Paper)}

%% To balance the two columns, you should reduce the text-height of
%% the last page using the following command:
%%%%%%%%%%%%%%%%%%%%%%%%%%%%%%%%%%%%%%%%%%%%%%%%%%%%%%%%%%%%%%%%%%%%%
%\addtolength{\textheight}{-9.35cm}
%%%%%%%%%%%%%%%%%%%%%%%%%%%%%%%%%%%%%%%%%%%%%%%%%%%%%%%%%%%%%%%%%%%%%
%% with an appropriate value. This command must be place on the second
%% last page, i.e., for a one-page abstract here, for a two-page
%% abstract right after the \maketitle command.

\maketitle
\begin{abstract}
We investigate the capacity of the $Q$-frequency $S$-user vector adder channel (channel with intensity information) introduced by Chang and Wolf.  Both coordinated and uncoordinated types of transmission are considered. Asymptotic (under the conditions $Q \to \infty$, $S = \gamma Q$ and $0 < \gamma < \infty$) upper and lower bounds on the relative (per subchannel) capacity are derived. The lower bound for the coordinated case is shown to increase when $\gamma$ grows. At the same time the relative capacity for the uncoordinated case is upper bounded by a constant.
\end{abstract}

\section{Introduction}
In \cite{CW} two multiuser channel models were introduced: the A-channel (or the channel without intensity information) and the B-channel (or the channel with intensity information). The capacity of the A-channel was investigated in \cite{CW, BP} for the case of coordinated transmission and in \cite{WZ, G, VK, VKK, GVC} for the case of uncoordinated transmission (the terminology is from \cite{CHV, V}). Note that the A-channel is in fact a vector disjunctive channel (OR channel) \cite{OFZ_2012, FZSF_2013}.

In this paper we investigate the capacity of the B-channel. The B-channel is a noiseless multiuser vector adder channel. Let us denote the number of active users by $S$, $S \ge 2$. For a certain time instant $\tau$ the channel inputs are binary vectors ${\bf{x}}_i^{(\tau)}, \: i = 1,2,\ldots, S$, of length $Q$ (the number of frequencies or subchannels) and of weight $1$ and the channel output at time instant $\tau$ is given by an elementwise sum of vectors at input
$$
{\bf y}^{(\tau)} = {\sum\limits_{i = 1}^{S}}{\bf x}^{(\tau)}_i.
$$
Note that the elements are added as real numbers.

The capacity of the B-channel for the coordinated case was investigated in \cite{CW} when $Q$ is fixed and $S \to \infty$. In this paper we are interested in the following asymptotics: $Q \to \infty$, $S = \gamma Q$ ($0 < \gamma < \infty$). If we take the limit as $Q \to \infty$, then the result of \cite{CW} corresponds to the case $\gamma \to \infty$. We also investigate the asymptotic capacity of the B-channel for the uncoordinated transmission, i.e. the type of transmission in which a user transmits the information independently of other users. This fact allows us to consider another users as noise. An uncoordinated transmission is preferable for
high-rate applications where a joint decoding is not possible for the complexity reasons.

Our contribution is as follows. Asymptotic (under the conditions $Q \to \infty$, $S = \gamma Q$ and $0 < \gamma < \infty$) upper and lower bounds on the relative (per subchannel) capacity are derived. The lower bound on the relative capacity for the coordinated case is shown to increase when $\gamma$ grows. At the same time the relative capacity for the uncoordinated case is upper bounded by a constant. The comparison with the result for the A-channel is done.

\section{Coordinated transmission}

Let us consider the case of coordinated transmission first. An example of a multiple-access system with coordinated transmission for a binary adder channel is given in \cite{ChWe}. Uniquely decodable codes are the major element of the system. Note that the system requires symbol and block synchronizations.

Let us denote by $X_i$ a vector sent by the $i$-th user ($i = 1, \ldots, S$) at a certain time instant, by $Y$ we denote the output of the channel at the time instant. The capacity (sum capacity) of the channel $C_{\text{c}}$ for the coordinated transmission is defined as follows
\begin{eqnarray*}
C_{\text{c}}(Q,S) &=& \max \left\{ I(X_1, X_2, \ldots, X_S; Y)\right\} \\
                  &=& \max \left\{ H(Y) - H(Y|X_1, X_2, \ldots, X_S)\right\} \\ 
                  &=& \max \left\{ H(Y)\right\},
\end{eqnarray*} 
where $H(X)$ is the binary entropy of a random variable, the maximum is taken over all possible {\it independent} distributions of random variables $X_1, X_2, \ldots, X_S$.

Since only the vectors of length $Q$ with the sum of elements equal to $S$ may be the channel outputs, then
\begin{equation}\label{C_coord_nonasympt_upper}
C_{\text{c}}(Q,S) \leq C^{\text{(U)}}_{\text{c}}(Q,S) = \log_2 \binom{S+Q-1}{S}, 
\end{equation}
as the number of such vectors is equal to $\binom{S+Q-1}{S}$.

In what follows we are interested in such an asymptotics: $Q \to \infty$, $S = \gamma Q$ ($0 < \gamma < \infty$). Let us introduce the notation of the asymptotic relative capacity
\[
{c_{\text{c}}(\gamma)} = \mathop{\lim}  \limits_{Q \to \infty} \left\{ C_{\text{c}}(Q, \gamma Q)/Q\right\}.
\]
The existence of the limit and the convexity of the function ${c_{\text{c}}(\gamma)}$ can be easily proved by corresponding frequency division (see \cite{BP}).

From (\ref{C_coord_nonasympt_upper}) we obtain
\[
{c_{\text{c}}(\gamma)} \leq c^{\text{(U)}}_{\text{c}}(\gamma) = (\gamma+1)\log_2(\gamma+1) - \gamma \log_2(\gamma).
\]

Now we derive a lower bound $c^{\text{(L)}}_{\text{c}}(\gamma)$. In \cite{CW} a formula for the entropy of the distribution at output $H(Y)$ is obtained when the all variables $X_1, X_2, \ldots, X_S$ are distributed uniformly, i.e.
\[
P(X_i = j) = \frac{1}{Q}, \quad i = 1, \ldots, S, \: j = 1, \ldots, Q,
\]
and an asymptotics of the quantity if found when $Q$ is fixed and $S \to \infty$. If we take the limit as $Q \to \infty$, we obtain that when $\gamma \to \infty$
\[
c^{\text{(L)}}_{\text{c}}(\gamma) \sim \frac{1}{2} \log_2 (2 \pi e \gamma).
\]

\begin{remark}
Here and in what follows by $P(X_i = j)$ we mean $P(X_i = \mathbf{x}_j)$, where $\mathbf{x}_j$ is a binary vector of length $Q$ with a single unit in the $j$-th position \textup(the positions are enumerated from $1$ to $Q$\textup).
\end{remark}

Let us consider the case when $\gamma$ is finite.

\begin{theorem}
Let $0 < \gamma < \infty$, then 
\[
{c_{\text{c}}(\gamma)} \geq c^{\text{(L)}}_{\text{c}}(\gamma) = \sum\limits_{i=0}^{\infty} \left[ \frac{\gamma^i}{i!} e^{-\gamma} \log_2(i!) \right] - \gamma \log_2 \left( \frac{\gamma}{e} \right).
\]
\end{theorem}
\begin{proof}
Let all the users use uniform distributions at input, then the probability to obtain the vector $\mathbf{y} = (y_1, y_2, \ldots, y_Q)$ at output of the channel can be calculated as follows
\begin{eqnarray*}
p( \mathbf{y} ) &=&  \binom{S}{y_1,y_2, \ldots, y_Q}  \left( \frac{1}{Q} \right)^S \\
                        &=& \frac{S!}{y_1! y_2!  \ldots  y_Q!}  \left( \frac{1}{Q} \right)^S.
\end{eqnarray*}

Thus,
\begin{flalign*}
&C^{\text{(L)}}_{\text{c}}(Q,S) = H(Y) = -\sum\limits_{\mathbf{y}} \left[ p(\mathbf{y}) \log_2 p(\mathbf{y})\right] \\
                                           &= -\sum\limits_{\mathbf{y}} \left[ p(\mathbf{y}) \log_2 \left( \binom{S}{y_1,y_2, \ldots, y_Q}  \left( \frac{1}{Q} \right)^S \right) \right] \\
                                           &= \sum\limits_{\mathbf{y}} \left[ p(\mathbf{y}) \sum\limits_{j=1}^{Q} \left\{\log_2 \left( y_j! \right)\right\}\right]  - \log_2\left( \frac{S!}{Q^S} \right)\\
                                           &= \sum\limits_{j=1}^{Q} \left\{ \sum\limits_{\mathbf{y}} \left[ p(\mathbf{y}) \log_2 \left( y_j! \right)\right] \right\} - \log_2\left( \frac{S!}{Q^S} \right)  \\
                                           &= Q \sum\limits_{i=0}^{S} \left[ \binom{S}{i} \left( \frac{1}{Q}\right)^i \left( 1-\frac{1}{Q}\right)^{S-i} \log_2 \left( i! \right)\right] \\
                                           &- \log_2\left( \frac{S!}{Q^S} \right),                                          
\end{flalign*}
the last transition is done in accordance to Lemma~\ref{lemma_sum_multinomial} (see the appendix).

After dividing by $Q$ and taking the limit as $Q \to \infty$ we obtain the needed result. 
\end{proof}

In \Fig{fig:coord} the derived bounds $c^{\text{(L)}}_{\text{c}}(\gamma)$ and $c^{\text{(U)}}_{\text{c}}(\gamma)$ are shown. For the comparison we also added the lower bound on the relative capacity $c^{\text{(disj)}}_{\text{c}}(\gamma)$ for the disjunctive channel (A-channel from \cite{CW}). The last bound was derived in \cite{BP}.

\begin{figure}[t]
\centering
\includegraphics[width=0.48\textwidth]{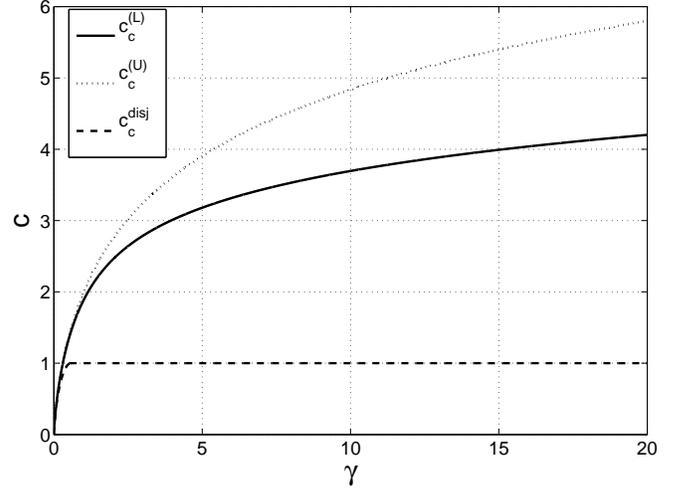}
\caption{Bounds on $c_{\text{c}}(\gamma)$}
\label{fig:coord}
\end{figure}

\section{Uncoordinated transmission}
Let us consider an uncoordinated transmission, i.e. the type of transmission where another users are considered as noise. The use of an uncoordinated transmission is preferable in the multiple-access systems with large number of active users with strict requirements to the transmission rate. An example of a multiple-access system with uncoordinated transmission for a disjunctive channel (OR channel) is given in \cite{CHV} and for a vector disjunctive channel in \cite{OFZ_2012, FZSF_2013}. Note that block synchronization is no more required.

In what follows we only consider the case when all the users use the same distributions of input symbols, i.e.
\[
P(X_i = j) = p_j, \quad i = 1, \ldots, S, \: j = 1, \ldots, Q.
\]
Note that this constraint is very natural for the uncoordinated transmission. 

The single-user capacity $C_i$ for the $i$-th user can be calculated as follows
\[
C_i(Q, S) = \max \left\{ I(X_i; Y) \right\},
\]
where the maximum is taken over all the distributions $X_i$.

Since all the users are ``equal'', then the capacity (sum capacity) $C_{\text{uc}}$ for the uncoordinated case can be calculated as a sum of single-user capacities
\[
C_{\text{uc}}(Q, S) = \sum\limits_{i=1}^{S} {C_i(Q, S)} = S \max\limits_{p_1, p_2, \ldots, p_Q} \left\{ I(X; Y) \right\}.
\]
In the last equality we used the fact that all the users use the same input distributions. For the same reason we dropped out the index $i$ in the notation of input $X$.

Analogously to the case of coordinated transmission we introduce the notation ($Q \to \infty$, $S = \gamma Q$)
\[
{c_{\text{uc}}(\gamma)} = \mathop{\lim}  \limits_{Q \to \infty} \left\{ C_{\text{uc}}(Q, \gamma Q)/Q\right\}.
\]
The proofs of the existence of the limit and of the convexity of the function ${c_{\text{uc}}(\gamma)}$ are little bit different here as all the users use the same distributions. We omit the proofs here.

\subsection{Upper bound}
It is clear that $C_{\text{uc}}(Q, S) \leq C_{\text{c}}(Q,S)$, then
\begin{equation} \label{upper_C_coord}
C_{\text{uc}}(Q, S) \leq \log_2 \binom{S+Q-1}{S}.
\end{equation}

Now we derive a stronger bound for large number of users.

\begin{theorem}\label{theorem_uc_upper}
The inequality holds
\[
C_{\text{uc}}(Q, S) \leq (Q-1) \log_2 e = (Q-1) 1.4427...
\]
\end{theorem} 
\begin{proof}
Note that
\[
p( \mathbf{y} ) =  \binom{S}{y_1,y_2, \ldots, y_Q}  p_1^{y_1} p_2^{y_2}  \ldots p_Q^{y_Q},
\]
\begin{flalign*}
&p( \mathbf{y} | \mathbf{x}_i ) =\\
&\left\{ {\begin{array}{rc}
  \binom{S-1}{y_1, \ldots, y_i -1, \ldots, y_Q}  p_1^{y_1} \ldots p_i^{y_i-1} \ldots p_Q^{y_Q}, & \: {\bf y}_i > 0, \\ 
  0,      & \: \text{otherwise.} 
\end{array}} \right.
\end{flalign*}

Thus,
\[
\frac{p( \mathbf{y} | \mathbf{x}_i )}{p( \mathbf{y} )} = \frac{y_i}{Sp_i}.
\]

\begin{flalign}\label{C_calc}
&I(X;Y) = \sum\limits_{\mathbf{x}} \sum\limits_{\mathbf{y}} \left[ p(\mathbf{x},\mathbf{y}) \log_2 \left( \frac{p( \mathbf{y} | \mathbf{x} )}{p( \mathbf{y} )} \right) \right]  \nonumber \\
&= \sum\limits_{j=1}^{Q} \sum\limits_{\mathbf{y}} \left[ p_j\binom{S-1}{y_1, \ldots, y_j -1, \ldots, y_Q} \right. \times \nonumber \\
&\times \left. p_1^{y_1} \ldots p_j^{y_j-1} \ldots p_Q^{y_Q} \log_2 \left( \frac{y_j}{Sp_j} \right)\right]  \nonumber \\
&= \sum\limits_{j=1}^{Q} \sum\limits_{i=0}^{S-1}\left[ p_j \binom{S-1}{i}p_j^i (1-p_j)^{S-1-i} \right. \times \nonumber \\
&\times \left. \log_2 \left( \frac{i+1}{Sp_j} \right) \right],
\end{flalign}
the last transition is done in accordance to Lemma~\ref{lemma_sum_multinomial} (see the appendix).

Applying the inequality
\[
\ln(1+x) \leq x,
\]
we obtain
\begin{flalign*}
&S I(X;Y)\\
&\leq \log_2e \sum\limits_{j=1}^{Q} \sum\limits_{i=0}^{S-1}\left[ \binom{S-1}{i} p_j^i (1-p_j)^{S-1-i} (i+1-Sp_j) \right] \\
&= \log_2e \sum\limits_{j=1}^{Q} \left[ (S-1)p_j+1-Sp_j) \right] = (Q-1) \log_2e,
\end{flalign*}
this completes the proof.
\end{proof}

From (\ref{upper_C_coord}) and Theorem~\ref{theorem_uc_upper} we obtain such an upper bound
\begin{eqnarray*}
{c_{\text{uc}}(\gamma)} &\leq&  c_{\text{uc}}^{\text{(U)}}(\gamma) \\
                                      &=& \min\left\{ (\gamma+1)\log_2(\gamma+1) - \gamma\log_2\gamma, \log_2 e \right\}.
\end{eqnarray*}

\begin{remark}
Sure the derived upper bound is not tight and can be improved. But already this rough bound shows that the quantity ${c_{\text{uc}}(\gamma)} $ is upper bounded by a constant.  
\end{remark}

\subsection{Lower bound}
In this section using several input distributions we obtain a lower bound on ${c_{\text{uc}}}$.

\subsubsection{Uniform distribution}

Let $p_1 = p_2 = \ldots = p_Q = 1/Q$.

\begin{statement}
Let $0 < \gamma < \infty$, then 
\[
{c_{\text{uc}}(\gamma)} \geq  c_{\text{uc}}^\text{unif}(\gamma) = \gamma \sum\limits_{i=0}^{\infty} {\frac{\gamma^i}{i!} e^{-\gamma} \log_2 \left(\frac{i+1}{\gamma}\right)}.
\]
\end{statement}
\begin{proof}
After substituting of the uniform distribution for (\ref{C_calc}) we obtain
\begin{flalign*}
&C_{\text{uc}}^\text{unif}(Q,S) \\
&= S \sum\limits_{i=0}^{S-1}\left[ \binom{S-1}{i} \left(\frac{1}{Q}\right)^i \left(1-\frac{1}{Q}\right)^{S-1-i} \log_2 \left( \frac{i+1}{\gamma} \right) \right].
\end{flalign*}

After dividing on $Q$ and taking the limit as $Q \to \infty$ we obtain the needed result. 
\end{proof}

\begin{figure}[t]
\centering
\includegraphics[width=0.48\textwidth]{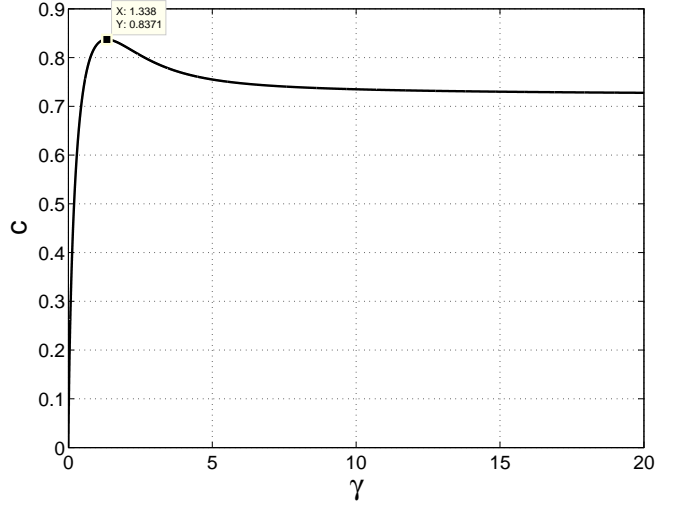}
\caption{The dependency  $c_{\text{uc}}^\text{unif}(\gamma)$}
\label{fig:unif}
\end{figure}

The dependency $c_{\text{uc}}^\text{unif}(\gamma)$ is shown in \Fig{fig:unif}. Let us introduce some notations
\begin{eqnarray*}
\gamma^* &=& \arg \max \limits_{\gamma} \left\{  c_{\text{uc}}^\text{unif}(\gamma) \right\}  = 1.3382\ldots \\
c_{\text{uc}}^* &=&  \max \limits_{\gamma} \left\{  c_{\text{uc}}^\text{unif}(\gamma) \right\} = 0.8371\ldots 
\end{eqnarray*}

Let us consider the case $\gamma \to \infty$.

\begin{statement}
\label{lim_sum_rate}
The equality follows
\[
\lim\limits_{\gamma \to \infty} \left\{ c_{\text{uc}}^\text{unif}(\gamma) \right\} = \frac{\log_2e}{2} = 0.7213\ldots
\]
\end{statement}
\begin{proof}
We need to use Lemma~\ref{lemma_limit} (see appendix).
\end{proof}

\subsubsection{Distorted distribution}

Let $S \geq \gamma^* (Q-1)$, we introduce the distorted distribution as follows
\[
\left\{ 
\begin{gathered}
  {p_1} = {p_2} =  \ldots  = {p_{Q - 1}} = \frac{{{\gamma ^*}}}{S} \hfill \\
  {p_Q} = 1 - (Q - 1)\frac{{{\gamma ^*}}}{S} \hfill \\ 
\end{gathered}  \right.
\]

\begin{statement}
Let $\gamma \geq \gamma^*$, then
\[
{c_{\text{uc}}(\gamma)} \geq  c_{\text{uc}}^*.
\]
\end{statement}
\begin{proof}
After substituting of the distorted distribution for (\ref{C_calc}) we obtain
\begin{flalign*}
&C_{\text{uc}}^\text{distort}(Q,S) \\
&= \gamma^* (Q-1) \sum\limits_{i=0}^{S-1}\left[ \binom{S-1}{i} \left(\frac{\gamma^*}{S}\right)^i \left(1-\frac{\gamma^*}{S}\right)^{S-1-i} \right.\\
& \times \left.\log_2 \left( \frac{i+1}{\gamma^*} \right) \right] \\
&+ S p_Q\sum\limits_{i=0}^{S-1}\left[ \binom{S-1}{i} \left(p_Q\right)^i \left(1-p_Q\right)^{S-1-i} \log_2 \left( \frac{i+1}{S p_Q} \right) \right].
\end{flalign*}

After dividing on $Q$ and taking the limit as $Q \to \infty$, we obtain
\begin{eqnarray*}
c_{\text{uc}}^\text{distort}(\gamma) &=& c_{\text{uc}}^* + f(\gamma),
\end{eqnarray*}
where                                                 
\begin{eqnarray*}
f(\gamma) &=& \lim\limits_{Q \to \infty} \left\{ \left( \gamma - \gamma^* \right) \sum\limits_{i=0}^{S-1}\left[ \binom{S-1}{i} \left(1 - \frac{\gamma^*}{\gamma}\right)^i \right.\right. \\
&\times& \left. \left. \left(\frac{\gamma^*}{\gamma}\right)^{S-1-i} \log_2 \left( \frac{i+1}{S \left( 1 - \frac{\gamma^*}{\gamma} \right)} \right) \right] \right\}.
\end{eqnarray*}

In accordance to Lemma~\ref{lemma_limit} (see appendix)
\[
f(\gamma) = \lim\limits_{Q \to \infty} \left\{\frac{\left( \gamma - \gamma^* \right)}{S}  \left( \frac{1}{2 \left( 1 - \frac{\gamma^*}{\gamma} \right)} + \frac{1}{2} \right) \log_2e\right\} = 0.
\]
\end{proof}

Thus we proved the following
\begin{theorem}
The inequality follows
\[
{c_{\text{uc}}(\gamma)} \geq c_{\text{uc}}^{\text{(L)}}(\gamma) =
\left\{ {\begin{array}{rc}
  c_{\text{uc}}^\text{unif}(\gamma), & \quad \gamma < \gamma^*, \\ 
  c_{\text{uc}}^*=0.8371...,      & \quad \gamma \geq \gamma^*. 
\end{array}} \right.
\]
\end{theorem}

\begin{figure}[t]
\centering
\includegraphics[width=0.48\textwidth]{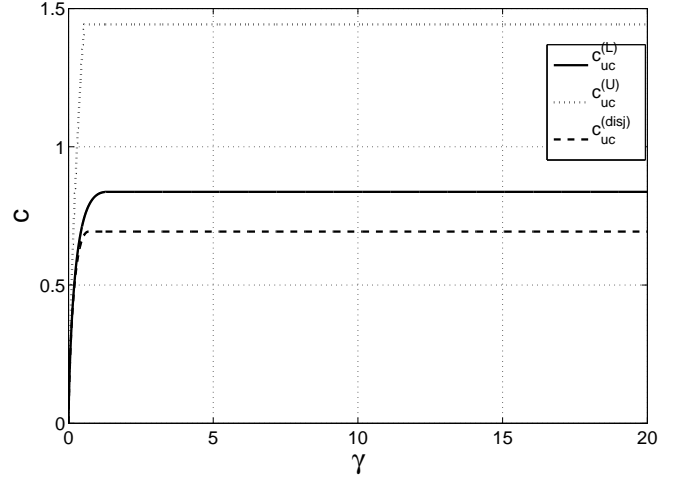}
\caption{Bounds on $c_{\text{uc}}(\gamma)$}
\label{fig:uncoord}
\end{figure}

In \Fig{fig:uncoord} the derived bounds $c^{\text{(L)}}_{\text{uc}}(\gamma)$ and $c^{\text{(U)}}_{\text{uc}}(\gamma)$ are shown, as for the coordinated case we add a lower bound $c^{\text{(disj)}}_{\text{uc}}(\gamma)$ on the capacity of the vector disjunctive channel (A-channel from \cite{CW}). The last bound is from \cite{WZ, G, VKK, VK}. One can see that the relative capacity for the uncoordinated case is upper bounded by a constant. We also note that the gain in comparison to the A-channel is not big.

\section{Conclusion}
Asymptotic (under the conditions $Q \to \infty$, $S = \gamma Q$ and $0 < \gamma < \infty$) upper and lower bounds on the relative (per subchannel) capacity are derived. The lower bound on the relative capacity for the coordinated case is shown to increase when $\gamma$ grows. At the same time the relative capacity for the uncoordinated case is upper bounded by a constant.

\section*{Acknowledgment}
We thank all the participants of the seminar on coding theory at IITP RAS, at which the paper was presented, and especially L.A.~Bassalygo and V.V.~Prelov. Without these people the paper would never have been written in the present form.

%For an appendix see comments in the template file.

%% Appendix:
%% If needed a single appendix is created by
%\appendix
%% If several appendices are needed, then the command
%\appendices
%% in combination with further \section-commands can be used.

%% Use \section* for acknowledgement
%\section*{Acknowledgment}

%The authors would like to thank various sponsors for supporting 
%their research. 

%% References:
%% We recommend the usage of BibTeX:
%%
%\bibliographystyle{IEEEtran}
%\bibliography{definitions,bibliofile}
%%
%% where we here have assume the existence of the files
%% definitions.bib and bibliofile.bib.
%% BibTeX documentation can be obtained at:
%% http://www.ctan.org/tex-archive/biblio/bibtex/contrib/doc/
%%
%%
%%
%% Or manual references (pay attention to consistency!):

\appendix

\begin{lemma}\label{lemma_sum_multinomial}
Let $p_i \geq 0$, \: $i = 1, \ldots, Q$, $\sum\nolimits_{i=1}^{Q} p_i = 1$ and $f(\cdot)$ be any function, then
\begin{flalign*}
&\sum\limits_{m_1 + \ldots + m_Q = S} \left[ \binom{S}{m_1, m_2, \ldots, m_Q} p_1^{m_1} p_2^{m_2} \ldots p_Q^{m_Q} f(m_1)\right] \\
&= \sum\limits_{i = 0}^{S} \left[ \binom{S}{i} p_1^{i} (1-p_1)^{S-i} f(i)\right].
\end{flalign*}
\end{lemma}
\begin{proof}
\begin{flalign*}
& \sum\limits_{m_1 + \ldots + m_Q = S} \left[ \binom{S}{m_1, m_2, \ldots, m_Q} p_1^{m_1} p_2^{m_2} \ldots p_Q^{m_Q} f(m_1)\right]  \\
&= \sum\limits_{i=0}^{S}\sum\limits_{m_2 + \ldots + m_Q = S-i} \left[ \binom{S}{i, m_2, \ldots, m_Q} \right. \\
&\times \left. p_1^{i} p_2^{m_2} \ldots p_Q^{m_Q} f(i)\right] \\
&= \sum\limits_{i=0}^{S} \left\{  \binom{S}{i} p_1^{i} f(i) \right. \\
& \times \left. \sum\limits_{m_2 + \ldots + m_Q = S-i} \left[ \binom{S-i}{m_2, \ldots, m_Q}  p_2^{m_2} \ldots p_Q^{m_Q} \right] \right\} \\
&= \sum\limits_{i=0}^{S} \left\{  \binom{S}{i} p_1^{i} (1-p_1)^{S-i} f(i)  \right\}.
\end{flalign*}
\end{proof}

\begin{lemma}\label{lemma_limit}
Let $0 < p < 1$, $N \to \infty$ and $pN \to \infty$, then
\[
\lim\limits_{N \to \infty} \left\{ N \sum\limits_{i = 0}^{N} { \left[ \binom{N}{i} {{p}^i}{{\left( {1 - p} \right)}^{N - i}}{\ln}\left( \frac{i + 1}{pN} \right) \right]} \right\} = \frac{1}{2p} + \frac{1}{2}.
\]
\end{lemma}
\begin{proof}
Let us consider the function
\[
G(p,N) = N \sum\limits_{i = 0}^{N} { \left[ \binom{N}{i} {{p}^i}{{\left( {1 - p} \right)}^{N - i}}{\ln}\left( \frac{i + 1}{pN} \right) \right]}.
\]

Let $\mu = pN$, $\varepsilon$ is an arbitrarily small positive value, let us divide the sum into three parts:
\[
G(p,N) = S_1 + S_2 + S_3,
\]
where
\begin{eqnarray*}
S_1 &=& N \sum\limits_{i = 0}^{(1-\varepsilon)\mu} { \left[ \binom{N}{i} {{p}^i}{{\left( {1 - p} \right)}^{N - i}}{\ln}\left( \frac{i + 1}{\mu}  \right) \right]},\\
S_2 &=& N \sum\limits_{i = (1-\varepsilon)\mu}^{(1+\varepsilon)\mu} { \left[ \binom{N}{i} {{p}^i}{{\left( {1 - p} \right)}^{N - i}}{\ln }\left( \frac{i + 1}{\mu} \right) \right]},\\
S_3 &=& N \sum\limits_{i = (1+\varepsilon)\mu}^{N} { \left[ \binom{N}{i} {{p}^i}{{\left( {1 - p} \right)}^{N - i}}{\ln}\left( \frac{i + 1}{\mu} \right) \right]}.
\end{eqnarray*}

In accordance to the Chernoff bound
\[
\lim\limits_{N\to\infty}{S_1} = \lim\limits_{N\to\infty}{S_3} = 0,
\]
and we only need to work with $S_2$.

Let us apply the following inequalities for logarithm ($-\varepsilon \leq x \leq \varepsilon $)
\[
\underline{L}(x) = x - \frac{x^2}{2} - \frac{\varepsilon^3}{3 (1-\varepsilon)} \leq \ln(1+x) \leq x - \frac{x^2}{2} + \frac{x^3}{3} = \overline{L}(x). 
\]

Let us consider two functions
\[
\underline{S_2}(p, N) = N \sum\limits_{i = (1-\varepsilon)\mu}^{(1+\varepsilon)\mu} { \left[ \binom{N}{i} {{p}^i}{{\left( {1 - p} \right)}^{N - i}} \underline{L} \left( \frac{i + 1 - \mu}{\mu} \right) \right]}
\]
and
\[
\overline{S_2}(p, N) = N \sum\limits_{i = (1-\varepsilon)\mu}^{(1+\varepsilon)\mu} { \left[ \binom{N}{i} {{p}^i}{{\left( {1 - p} \right)}^{N - i}} \overline{L}\left( \frac{i + 1 - \mu}{\mu} \right) \right]},
\]
it is clear, that $\underline{S_2}(p, N) \leq S_2(p, N) \leq \overline{S_2}(p, N)$.

Consider $\lim\nolimits_{N\to\infty}{\overline{S_2}(p, N)}$. It is easy to check, that
\begin{flalign*}
&\lim\limits_{N\to\infty} \left\{ N \sum\limits_{i = (1-\varepsilon)\mu}^{(1+\varepsilon)\mu} { \left[ \binom{N}{i} {{p}^i}{{\left( {1 - p} \right)}^{N - i}}  \left( \frac{i + 1 - \mu}{\mu} \right) \right]} \right\} \\
&= \frac{1}{p}, \\
&\lim\limits_{N\to\infty} \left\{ N \sum\limits_{i = (1-\varepsilon)\mu}^{(1+\varepsilon)\mu} { \left[ \binom{N}{i} {{p}^i}{{\left( {1 - p} \right)}^{N - i}} \frac{1}{2} \left( \frac{i + 1 - \mu}{\mu} \right)^2 \right]} \right\} \\
&= \frac{1-p}{2p}, \\
&\lim\limits_{N\to\infty} \left\{ N \sum\limits_{i = (1-\varepsilon)\mu}^{(1+\varepsilon)\mu} { \left[ \binom{N}{i} {{p}^i}{{\left( {1 - p} \right)}^{N - i}}  \frac{1}{3} \left( \frac{i + 1 - \mu}{\mu} \right)^3 \right]} \right\} \\
&= 0.
\end{flalign*}

Thus, 
\[
\lim\nolimits_{N\to\infty}{\overline{S_2}(p, N)} = \frac{1}{2p} + \frac{1}{2}.
\]
Similarly 
\[
\lim\nolimits_{N\to\infty}{\underline{S_2}(p, N)} = \frac{1}{2p} + \frac{1}{2} - \frac{\varepsilon^3}{3 (1- \varepsilon)}.
\]

As $\varepsilon$ can be chosen arbitrarily small, then
\[
\lim\limits_{N\to\infty}{S_2(p, N)} = \frac{1}{2p} + \frac{1}{2}.
\] 
\end{proof}

\end{document}